\documentclass[onecolumn, 12 pt, draftcls]{ieeetran}

\usepackage[dvips]{graphicx}
\usepackage[cmex10]{amsmath}
\usepackage{amssymb}
\usepackage{array}
\usepackage{enumerate}
\usepackage{psfrag}
\usepackage{stfloats}
\usepackage{dsfont}

\newtheorem{thm}{Theorem}
\newtheorem{lem}{Lemma}

\newtheorem{defn}{Definition}

\begin{document}

\title{Joint Source-Channel Cooperative Transmission over Relay-Broadcast Networks}
\author{%
  \authorblockN{Deniz G\"{u}nd\"{u}z\authorrefmark{1},
    Elza Erkip\authorrefmark{2},
    Andrea Goldsmith\authorrefmark{3}, H. Vincent Poor\authorrefmark{4}
  }\\
  \authorblockA{\small
    \authorrefmark{1}CTTC, Barcelona, Spain\\
  }
  \authorblockA{%
    \authorrefmark{2}Dept. of Electrical Eng., Polytechnic Institute of New York University, Brooklyn, NY\\
  }
  \authorblockA{%
    \authorrefmark{3}Dept. of Electrical Eng., Stanford Univ., Stanford, CA\\
  }
  \authorblockA{%
    \authorrefmark{4}Dept. of Electrical Eng., Princeton Univ., Princeton, NJ\\
  }

  Email: deniz.gunduz@cttc.es, elza@poly.edu, andrea@wsl.stanford.edu, poor@princeton.edu
   \thanks{This work is supported in part by a Marie Curie grant funded by European
Union's Seventh Framework Programme (FP7), by the Spanish Government under project TEC2010-17816 (JUNTOS), by the National Science Foundation under Grant CCF-0635177, by the DARPA ITMANET program under grant 1105741-1-TFIND and by the ARO under MURI award W911NF-05-1- 0246.}
}

\date{}
\maketitle

\begin{abstract}
Reliable transmission of a discrete memoryless source over a multiple-relay relay-broadcast network is considered. Motivated by sensor network applications, it is assumed that the relays and the destinations all have access to side information correlated with the underlying source signal. Joint source-channel cooperative transmission is studied in which the relays help the transmission of the source signal to the destinations by using both their overheard signals, as in the classical channel cooperation scenario, as well as the available correlated side information. Decode-and-forward (DF) based cooperative transmission is considered in a network of multiple relay terminals and two different achievability schemes are proposed: i) a regular encoding and sliding-window decoding scheme without explicit source binning at the encoder, and ii) a semi-regular encoding and backward decoding scheme with binning based on the side information statistics. It is shown that both of these schemes lead to the same source-channel code rate, which is shown to be the \textit{source-channel capacity} in the case of i) a physically degraded relay network in which the side information signals are also degraded in the same order as the channel; and ii) a relay-broadcast network in which all the terminals want to reconstruct the source reliably, while at most one of them can act as a relay.
\end{abstract}

\section{Introduction}\label{s:introduction}

A relay network consists of a source-destination pair and dedicated relay terminals that help the transmission of messages from the source to the destination. The classical relay channel model \cite{Cover:IT:79} focuses on the maximum channel coding rate that can be achieved with arbitrarily small probability of error. Most studies on the relay channel following \cite{Cover:IT:79} focus solely on the channel coding aspects of relaying, motivated by the improvement in the capacity, reliability or coverage extension provided by the relay terminals. On the other hand, in some applications, such as sensor networks, the relays might have partial information about the source signal obtained through their own sensing capabilities. This additional side information can be used to improve the end-to-end system performance. The advantages of exploiting the correlated side information in a sensor network from a purely source coding perspective have been illustrated in \cite{Cristescu:IT:05}.

Here we consider the transmission of a discrete memoryless (DM) source over a DM relay-broadcast network with multiple relays and destinations, in which all the terminals in the network have access to their own correlated side information. The relays in the network are dedicated terminals whose only goal is to help the transmission of the source signal to the destinations. Some of the destinations also have channel inputs through which they can also relay the source signal to each other. The goal is the reliable (lossless) transmission of the underlying source signal to the destination(s), and the problem is to characterize the maximum number of source symbols per channel use\footnote{Here, the ``channel use'' refers to the use of the whole network, not the use of the separate source-relay or relay-destination channels. A more appropriate term would be the ``network use'', but we stick here to the more common terminology in the literature.} that can be transmitted reliably, called the \textit{source-channel capacity}. This is a joint source-channel coding generalization of the classical relay network problem.

In this model the transmission scheme should exploit the availability of the side information at the network terminals as well as the overheard channel transmissions. Note that the classical channel cooperation ignores the side information at the terminals. However, this can lead to a significant performance loss. Consider, for example, a single relay channel in which there is no channel from the source terminal to the relay, i.e., the relay channel output is independent of the source terminal's channel input, while the relay has access to side information correlated with the underlying source signal. In this case the relay can still cooperate with the source terminal by forwarding its side information to the destination. This special case is called the \textit{one-helper problem}, and the benefits of cooperation in this setup are shown in \cite{Behroozi:ICC:07} and \cite{Aguerri:Globecom:10}. Classical channel cooperation schemes cannot exploit the side information at the relay terminal, since their focus is on the processing of the overheard signal received by the relay.

Several channel coding techniques have been proposed for the relay channel \cite{Cover:IT:79}. In the decode-and-forward (DF) protocol the relays decode the underlying message, and cooperate with the source terminal to forward it to the destination. While not optimal in general, DF achieves the capacity in a physically degraded relay channel \cite{Cover:IT:79}. Here, we focus on the DF protocol in the joint source-channel transmission setting and propose multiple-relay extensions that exploit the side information at the relays and the destinations.

The DF protocols in the literature are categorized based on the codebook sizes and the decoding strategy. In \emph{irregular encoding and successive decoding} \cite{Cover:IT:79}, the relay and the source codebooks have different sizes and the destination applies successive decoding. In \textit{regular encoding and sliding-window decoding}, introduced in \cite{Carleil:IT:82}, the source and the relay codebooks have the same size and the destination decodes each source message by using two consecutive channel blocks. Finally, in \textit{regular encoding and backward decoding}, introduced in \cite{Willems:thesis}, the destination waits until all channel blocks are received, and decodes the messages starting from the last block and going backwards. For single source-single destination relay networks all encoding schemes lead to the same set of achievable rates despite having different delay behaviors; however, this is not the case in general, and backward decoding might lead to higher rates when there are multiple source terminals in the network \cite{Laneman:ISIT:04}. DF channel coding is extended to multiple-relay networks in  \cite{Kramer:IT:05}, \cite{Aref:thesis}, \cite{Xie:IT:05}, \cite{Sikora:IT:06} and \cite{Razaghi:IT:09}. While \cite{Aref:thesis} and \cite{Razaghi:IT:09} consider irregular encoding, \cite{Kramer:IT:05} and \cite{Xie:IT:05} study an extension of the regular encoding and sliding-window decoding scheme, and finally \cite{Sikora:IT:06} and \cite{Kramer:Allerton:05} extend the backward decoding strategy to multiple relays.

We propose two different \textit{joint source-channel cooperation} protocols based on DF relaying. In particular we consider the joint source-channel coding extensions of the sliding-window \cite{Kramer:IT:05}, \cite{Xie:IT:05} and backward decoding schemes \cite{Sikora:IT:06}, \cite{Kramer:Allerton:05}. These two transmission strategies differ in terms of the sizes of the codebooks used and the decoding delays, as well as the source encoding techniques; hence, these two schemes offer a tradeoff between decoding delay and complexity in the system, while achieving the same performance in terms of the source samples that can be transmitted per channel use, i.e., the source-channel code rate.

In transmitting a source signal to another user with correlated side information, the classical random coding approach developed by Slepian and Wolf in \cite{Slepian:1973} is to use \emph{binning} to reduce the amount of transmitted information. Source binning, which will be defined more rigorously in Section \ref{s:PtP}, refers to dividing the possible source sequences into groups and sending only the index of the group rather than the index of the sequence. The decoder can then decode the source sequence using its side information together with the bin index. In the relay network setup with the DF protocol, due to the varying quality of side information at the terminals, separate binning is required for each terminal, and the corresponding bin indices are transmitted with channel codes at different rates. This will be the approach taken in the construction of the backward decoding scheme. However; we will see that transmission at the same rate is also possible without resorting to any binning operation. This will be the approach for building the sliding-window decoding scheme.

We should also remark that the proposed protocols are not expected to achieve the optimal performance in the general setting since our problem is a generalization of the classical relay network problem, which remains open. However, we prove that the proposed DF-based protocols achieve the optimal source-channel code rate, i.e., the source-channel capacity, in a physically degraded setting in which both the channel outputs and the side information sequences are degraded in the same order, and in a relay-broadcast network with one source and multiple destinations, such that at most one of the destinations can also act as a relay.

The problem of joint source-channel cooperative transmission has been previously studied for a single relay channel in \cite{Gunduz:ITW:07}, \cite{Kwak:ISIT:07} and \cite{Smith:Allerton:04}, and for a multiple-relay network in \cite{Sefidgaran:ISIT:09} and \cite{Gunduz:CrownCom:10}. The techniques proposed in all these works are based on DF relaying with different transmission techniques. While semi-regular encoding and backward decoding with explicit binning at the source encoder is proposed in \cite{Gunduz:ITW:07}, irregular encoding/ successive decoding with and without explicit binning is considered in \cite{Kwak:ISIT:07} and \cite{Smith:Allerton:04}, respectively. In \cite{Sefidgaran:ISIT:09} a regular encoding/ sliding-window decoding scheme with explicit binning is considered in the multiple-relay setting. The lossy version of joint source-channel cooperative transmission is studied in \cite{Gunduz:ISIT:08b}.

The rest of the paper is organized as follows: In Section \ref{s:system_model}, we
introduce the system model and the problem. In Section \ref{s:PtP} we illustrate the difference between binning and non-binning based joint source-channel coding schemes in a point-to-point scenario. The main results of the paper are stated in Section \ref{s:results}, in which the focus is on a relay network with a single destination terminal. In Section \ref{s:relay_bc} we extend the results to relay broadcast networks, and show that the proposed achievability technique based on joint source-channel DF scheme is optimal when there is only one destination with transmission capability. The paper is concluded in Section \ref{s:conc}. The proofs of our main results are detailed in the Appendices.

In this paper we denote random variables by capital letters, sample values by the respective lower case letters, and alphabets by the respective calligraphic letters. The cardinality of set $\mathcal{A}$ is denoted by $|\mathcal{A}|$. For $k \leq n$, the sequence $(X_k,\ldots,X_n)$ will be denoted by $X_k^n$, while $X^n$ will be used for $X_1^n$. The complement of a certain element $X_i$ in a vector $X^n$ will be denoted by $X_i^c \triangleq (X_1, \ldots, X_{i-1}, X_{i+1}, \ldots,X_n)$. 

Let $\pi(\cdot)$ be an injective function\footnote{A function $f$ on a set $\mathcal{A}$ is injective if for all $a, b \in \mathcal{A}$, if $f(a) = f(b)$, then $a = b$; that is, $f(a) = f(b)$ implies $a = b$.} from a set $\{1, \ldots, N\}$ to set $\{1, \ldots, K\}$ with $1 \leq N \leq K$. For $1 \leq i<j \leq N$, we define $\pi(i:j) \triangleq \{\pi(i), \pi(i+1), \ldots, \pi(j)\}$. We also define, for a set $C=\{c_1, \ldots, c_n\}$, $n\in \mathds{Z}^+$ and $c_i\in \mathds{Z}^+$, $X_{C}\triangleq (X_{c_1}, \ldots, X_{c_n})$.

\section{Problem Setup}\label{s:system_model}

We first consider the relay network with a single source-destination pair and multiple relays. We have a network of $K+2$ terminals (see Fig. \ref{f:system_model}): terminal $T_0$ is the source terminal observing the source signal $S_0$, terminals $T_i$ for $i=1, \ldots, K$ are the $K$ relay terminals each observing a different correlated side information signal $S_i$, and terminal $T_{K+1}$ is the destination terminal with its own correlated side information signal $S_{K+1}$. The underlying DM relay channel is characterized by the conditional distribution
\begin{align*}
p & (y_1^n, \ldots, y_{K+1}^n|x_0^n, x_1^n, \ldots, x_{K}^n) = \prod_{t=1}^n
p_{Y_1,\ldots,Y_{K+1}|X_0,\ldots,X_{K}}(y_{1,t},\ldots, y_{K+1,t}|x_{0,t}, \ldots, x_{K,t}),
\end{align*}
where $x_{i,t} \in \mathcal{X}_i$ and $y_{i,t} \in \mathcal{Y}_i$, respectively, are the channel input and output of terminal $T_i$ at time $t$; and the finite sets $\mathcal{X}_i$ and $\mathcal{Y}_i$ are the corresponding input and output alphabets. We denote this channel by $(\mathcal{X}_0, \ldots, \mathcal{X}_{K}, p(y_1^n, \ldots, y_{K+1}^n|x_0^n, x_1^n, \ldots, x_{K}^n), \mathcal{Y}_1, \ldots, \mathcal{Y}_{K+1})$.

We consider DM independent and identically distributed (i.i.d.) signals $(S_0,\ldots,S_{K+1})$ which are arbitrarily correlated according to a joint distribution
$p(s_0,\ldots,s_{K+1})$ over a finite alphabet $\mathcal{S}_0 \times
\cdots \times \mathcal{S}_{K+1}$. The sequence $\{S_{0,j}\}_{j=1}^\infty$ is denoted as the source sequence while $\{S_{i,j}\}_{j=1}^\infty$, $i=1, \ldots, K+1$, is the side information sequence available at terminal $T_i$. We denote the set composed of the source and the side information signals by $(\mathcal{S}_0, \ldots, \mathcal{S}_{K+1}), p(s_0,\ldots,s_{K+1})$.

\begin{figure*}
\centering
\small
\psfrag{X0}{$X_0$}\psfrag{X1}{$X_1$}\psfrag{Xk}{$X_K$}
\psfrag{Y1}{$Y_1$}\psfrag{Yk}{$Y_K$}\psfrag{Yk1}{$Y_{K+1}$}
\psfrag{S}{$\hspace{-0.1in}S_0$}\psfrag{S1}{$S_1$}\psfrag{Sk}{$S_K$}\psfrag{Sk1}{$S_{K+1}$}
\psfrag{hS}{$\hat{S}_0$}
\psfrag{Source}{$\mathrm{Source,T_0}$}
\psfrag{Kanal}{$\mathrm{Channel}$} \psfrag{R1}{$\mathrm{Relay~T_1}$}\psfrag{Rk}{$\mathrm{Relay~T_K}$}
\psfrag{D}{$\mathrm{Destination,~T_{K+1}}$}
\psfrag{Ch}{$p(y_1, \ldots, y_{K+1}|x_0, x_1, \ldots, x_{K})$}
\includegraphics[width=6.2in]{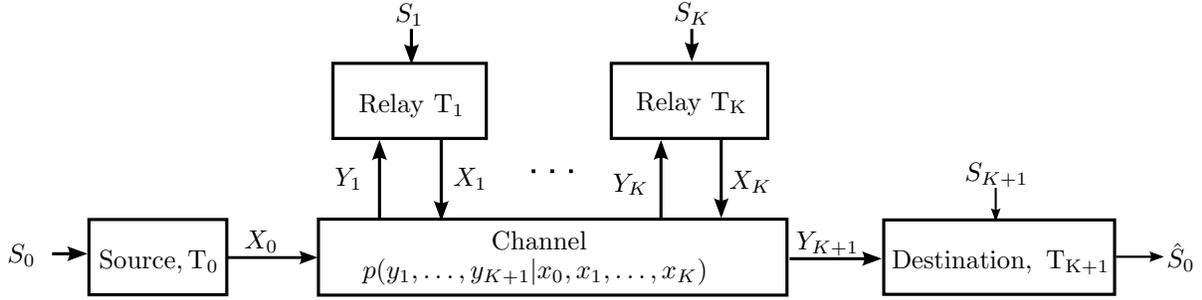}
\caption{Transmission of a discrete memoryless source over a relay network with correlated side information.} \label{f:system_model}
\end{figure*}

\begin{defn}
An $(m,n)$ joint source-channel code with the sources $(\mathcal{S}_0, \ldots, \mathcal{S}_{K+1})$ and the relay network $(\mathcal{X}_0, \ldots, \mathcal{X}_{K}, p(y_1^n, \ldots, y_{K+1}^n|x_0^n, x_1^n, \ldots, x_{K}^n), \mathcal{Y}_1, \ldots, \mathcal{Y}_{K+1})$ consists of
\begin{enumerate}
  \item An encoding function $f_0^{(m,n)}:\mathcal{S}^m \rightarrow \mathcal{X}_0^n$ at the source terminal $T_0$ that maps its observation $S_0^m$ to a channel codeword of length-$n$, i.e., $X_0^n=f_0^{(m,n)}(S_0^m)$.
  \item A series of encoding functions for each relay terminal $T_i$, $i=1, \ldots, K$: $f_i^{(m,n)} = \{ f_{i,1}^{(m,n)}, \ldots, f_{i,n}^{(m,n)}\}$, such that the encoding function at time instant $t$ depends on the previous channel outputs $Y_i^{t-1}$ as well as its side information vector $S_i^m$. We have
    \[
        X_{i,t} = f_{i,t}^{(m,n)}(Y_{i,1}, \ldots, Y_{i, t-1}, S_i^m),
    \]
    for $1\leq i\leq K$ and $1\leq t \leq n$.
  \item A decoding function at the destination terminal $T_{K+1}$ which maps the channel output $Y_{K+1}^n$ and its side information $S_{K+1}^m$ to the estimate $\hat{S}_0^m$ by
    \begin{eqnarray}
        g^{(m,n)}: \mathcal{Y}_{K+1}^n \times \mathcal{S}_{K+1}^m \rightarrow \mathcal{S}_0^m,
    \end{eqnarray}
i.e., $\hat{S}_0^m = g^{(m,n)}(Y_{K+1}^n, S_{K+1}^m)$.
\end{enumerate}
\end{defn}

The goal of the network is to transmit the source message $S_0^m$ to the destination terminal in a reliable manner. Reliability is based on the following definition of average probability of error.

\begin{defn}
The probability of error for an $(m,n)$ code is defined as
\[
    P_e^{(m,n)} = Pr\left\{ \hat{S}_{K+1}^m \neq S_0^m \right\},
\]
where the averaging is over both the source and the channel distributions.
\end{defn}

\begin{defn}\label{d:achieve}
We say that the source-channel code rate of $r$ source samples per channel use is \emph{achievable} if there exists a sequence of $(m,n)$ codes satisfying $\frac{m}{n} \geq r$ whose probability of error vanishes with increasing block size, i.e., $P_e^{(m,n)} \rightarrow 0$ as $m,n \rightarrow \infty$.
\end{defn}

In \cite{Gunduz:IT:09} we have defined the \textit{source-channel rate} as the number of channel uses required to transmit each source sample, that is, as the inverse of the \textit{source-channel code rate} we have defined here. The latter definition is in accordance with the definition of the rate of a joint source-channel code used in \cite{Gastpar:IT:03}. Similar to the definition of channel capacity, we define a source-channel capacity of a network which considers both the source and the channel characteristics.

\begin{defn}\label{d:achieve}
The \emph{source-channel capacity} of a network is defined as the supremum of all achievable source-channel code rates.
\end{defn}

\begin{figure*}
\centering
\small
\psfrag{X}{$X_0$}\psfrag{Y}{$Y_1$}
\psfrag{S}{$\hspace{-0.1in}S_0$}\psfrag{S2}{$S_1$}
\psfrag{hS}{$\hat{S}_0$}
\psfrag{Source}{$\mathrm{Source}$}
\psfrag{Ch}{$p(y_1|x_0)$}
\psfrag{D}{$\mathrm{Destination}$}
\includegraphics[width=4in]{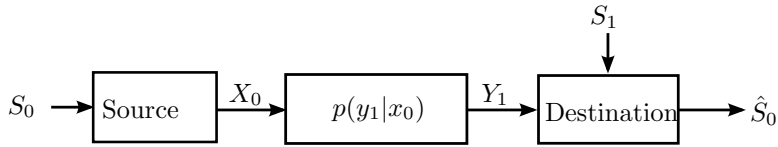}
\caption{Transmission of a discrete memoryless source over a point-to-point channel with correlated side information.} \label{f:PtP}
\end{figure*}

\section{Point-to-point Channel: To Bin or Not To Bin}\label{s:PtP}

In this section we focus on the point-to-point channel setup to introduce the basic concepts such as binning and source-channel separation, which will be helpful in understanding the coding schemes proposed for the relay network. The point-to-point channel without side information at the receiver was studied by Shannon, who proved the optimality of source-channel separation in this setup \cite{Shannon:1948}. The point-to-point channel with side information at the receiver was studied by Shamai and Verd\'{u} in \cite{Shamai:ETT:95}. In the system model introduced in Section \ref{s:system_model}, the point-to-point channel with receiver side information corresponds to the case with $K=0$ (see Fig. \ref{f:PtP}). It is shown in \cite{Shamai:ETT:95} that a source-channel code rate $r$ is achievable if there exists an input distribution $p(x_0)$ such that $r < \frac{I(X_0;Y_1)}{H(S_0|S_1)}$; and conversely, if the source-channel code rate $r$ is achievable then there exists an input distribution $p(x_0)$ such that $r \leq \frac{I(X_0;Y_1)}{H(S_0|S_1)}$. Note that $I(X_0;Y_1)$ is maximized by the capacity achieving input distribution. Hence, equivalently, the source-channel capacity of this system is $\frac{C}{H(S_0|S_1)}$, where $C$ is the channel capacity. Moreover, it is not very difficult to see that the source-channel separation theorem holds for this setup, that is, any source-channel code rate less than the source-channel capacity, $\frac{C}{H(S_0|S_1)}$ can be achieved by first applying source coding, and then transmitting the compressed source bits over the channel using a capacity achieving channel code. As pointed out in \cite{Shamai:ETT:95}, the source encoder in the case of separation is a Slepian-Wolf encoder that operates at the conditional entropy rate of $H(S_0|S_1)$ rather than the source entropy rate of $H(S_0)$, due to the availability of the correlated side information at the receiver.

Slepian-Wolf compression involves \emph{binning} of the source outcomes. The source encoder randomly distributes all possible source output sequences $S_0^m$ into $2^{mH(S_0|S_1)}$ bins, that is, it independently assigns an index uniformly distributed over $\{1, \ldots, 2^{mH(S_0|S_1)} \}$ to each of the possible source output sequences. The bin index for each source sequence is transmitted over the channel by assigning a channel codeword to each bin index. Having decoded the channel codeword, and hence, the bin index, correctly, the receiver outputs the source outcome in the corresponding bin that is jointly typical with its side information sequence.

We next present a coding scheme that generalizes this separate source-channel coding approach. In this generalized scheme, we randomly distribute all $S_0^m$ sequences into $M=2^{mR}$ bins, where $R$ is not necessarily equal to $H(S_0|S_1)$. Let $\mathcal{B}(i)$ be the set of sequences allocated to bin $i$. Then, we generate $M$ independent channel codewords of length $n$ according to distribution $\prod_{t=1}^m p(x_{0,t})$, and enumerate these codewords as $x_0^n(w)$ for $w=1, \ldots, M$. This constitutes the only codebook in the system. Encoding is done as in the separation scheme. The transmitter finds the index $i$ of the bin to which $s_0^m$ belongs, and transmits over the channel the codeword $x_0^n(i)$.

In the classical source-channel separation approach, the channel decoder, upon observing the channel output, decides a single codeword index, and conveys this index to the source decoder. This index corresponds to the bin index to which the source sequence belongs.  Then the source decoder estimates the source sequence using the bin index and the side information sequence. For the channel transmission to be successful with high probability, the rate of transmission should be less than the channel capacity. Hence, in the separation approach, it is advantageous to reduce the number of possible indices to be transmitted through source binning to achieve higher source-channel code rates. However, in the generalized scheme, we consider a joint source-channel decoder, following the approach in \cite{Tuncel:IT:06}. The decoder at the receiver looks for an index $i$ for which $x_0^n(i)$ and $Y_1^n$ are jointly typical, and at the same time, there exists exactly one source sequence $\hat{s}^m$ in bin $i$ that is jointly typical with $S_1^m$. This source sequence is the output of the decoder. Note that, due to the joint nature of decoding using both the side information and the channel output, this is not separate source and channel coding in the strict sense. We refer the readers to \cite{Tuncel:IT:06} and \cite{Gunduz:IT:09} for more discussion on source-channel separation in multi-terminal scenarios.

We have an error if there exists no or more than one such bin index $i$, or if there exists more than one jointly typical sequence within bin $i$. The probability that there is no bin index satisfying the joint typicality condition vanishes as $n$ grows. The probability of having no jointly typical source sequence within the correct bin also vanishes since $S_0^m$ and $S_1^m$ are jointly typical with high probability as $m$ grows. The probability of having another jointly typical source sequence in the same bin as $S_0^m$ is bounded by
\begin{align}\label{eq:PtP:1}
\big| \mathcal{B}(i) \bigcap A_\epsilon^m(S_0) \big| 2^{-m(I(S_0; S_1)-3\epsilon)} \leq 2^{m(H(S_0)+\epsilon)} 2^{-mR} 2^{-m(I(S_0; S_1)-3\epsilon)},
\end{align}
in which $A_\epsilon^m(S_0)$ denotes the set of $\epsilon$-typical $n$-tuples according to $P_{S_0}$. Using the classical arguments on typical sets \cite{Csiszar:book}, it is possible to show that (\ref{eq:PtP:1}) goes to zero if $R \geq H(S_0|S_1)$.

We also have an error if there exists another bin index $j$ satisfying the joint typicality conditions. The probability of this event can be bounded by
\[
2^{mR} 2^{-n(I(X_0; Y_1)-3\epsilon)} |\mathcal{B}(i) \bigcap A_\epsilon^m(S_0)|2^{-m(I(S_0; S_1)-3\epsilon)} \leq 2^{-n(I(X_0; Y_1)-3\epsilon)} 2^{m(H(S_0|S_1)-2\epsilon)},
\]
which goes to zero if $mH(S_0|S_1) < nI(X_0 ; Y_1)$. Hence, any rate $r$ satisfying $r <\frac{I(X_0; Y_1)}{H(S_0|S_1)}$ is achievable.

Now, we have a set of coding schemes each with a different number of source bins, that is, with different $R$ values satisfying $R \geq H(S_0|S_1)$. As suggested in \cite{Laneman:IT:05}, the ``joint'' decoding operation considered in the generalized scheme can equivalently be viewed as a separate source and channel decoding scheme, in which the channel decoder is a list decoder, which outputs the list of bin indices $i$ for which $x_0^n(i)$ and $Y_1^n$ are jointly typical. This list decoding approach includes separate source-channel coding as a special case with $R=H(S_0|S_1)$; in which case we have, with high probability, a single element in the list, i.e., there exists only a single bin index whose channel input codeword is typical with the channel output.

We want to point out here that, on the other extreme, this generalized scheme works without any binning, such that we generate an independent channel codeword for each possible source outcome, i.e., $R = \log |\mathcal{S}_0|$. From a practical point of view, this can be seen as transferring the complexity of binning from the encoder to the decoder, which now needs to apply joint decoding or list decoding. From a theoretical point of view, since the decoder only outputs typical source sequences as its estimate, there is no point in having more than $2^{m(H(S_0)+\epsilon)}$ bins as, otherwise, there would be bins without any typical source sequences, i.e., we have $R \leq H(S_0)$. Hence, in the rest of the paper, schemes with $R =H(S_0)$ are considered as \emph{no-binning} schemes.

In the case of a point-to-point channel, the only difference between separate source and channel coding with binning, and joint decoding with no-binning is the operation at the encoder and the decoder. However, as we will see in the following sections, in the case of relay networks these two approaches require different transmission protocols, and hence, obtain different delay performances even though they still achieve the same source-channel code rate performance.

\section{Main Results}\label{s:results}

In this section we address the main problem of source transmission over relay networks. We state our results for the single destination setup in this section while the proofs are given in the following sections. The first theorem provides an achievability result.

\begin{thm}\label{t:suff_cond}
For the DM relay network with correlated relay and destination side information, the source-channel code rate $r$ is achievable if, for all $i = 1,\ldots, N$,
\begin{align}
    r < \frac{I(X_{\pi(0:i-1)};Y_{\pi(i)} |X_{\pi(i:N-1)})}{H(S_0|S_{\pi(i)})}, \label{cond1}
\end{align}
for an $N$ satisfying $1 \leq N \leq K+1$, and for some injection $\pi(\cdot)$ from $\{0,\ldots, N\}$ to $\{0,\ldots, K+1\}$ such that $\pi(0)=0$, $\pi(N)= K+1$, and for some input distribution $p(x_0,\ldots, x_{K})$ and $\mathcal{X}_{K+1}$ is a constant.
\vspace{.1in}
\end{thm}

Note in (\ref{cond1}) that we impose decoding constraints for only a subset of the relays in the network together with the destination, such that only $N-1$ relays decode rather than all $K$ of them. Excluding some of the relays from cooperation, the relays with a poor channel from the source terminal as well as poor side information, can potentially increase the rate in the case of DF relaying, as these relays can decode only at very low rates.


In the appendix we provide two different proofs for the achievability of Theorem \ref{t:suff_cond}. Both proofs are based on DF relaying in the joint source-channel setting, that is, the source vector $S_0^m$ is decoded in a lossless fashion by all the terminals participating in cooperation. The first proof is based on block-Markov regular encoding and sliding-window decoding without explicit binning. This achievable scheme is based on the joint source-channel decoding scheme with no binning introduced in Section \ref{s:PtP}. Here, we combine this technique with block Markov encoding in a multiple-relay setup. The typical source outcomes are mapped directly to different channel codewords rather than binning the source outputs prior to channel coding\footnote{This is equivalent to binning all source outcomes into $R = 2^{mH(S_0)}$ bins in the scheme introduced in Section \ref{s:PtP}, so that there is one typical source outcome in each bin.}. The decoders at each relay apply joint source-channel decoding (or, separate  channel and source decoders in which the channel decoder outputs a list of possible inputs rather than a single codeword index). Each relay finds the unique index for which the corresponding source codeword is jointly typical with its side information while the corresponding channel codewords are jointly typical with the received channel vectors in the preceding blocks. This is a regular coding scheme since all the terminals in the network use a codebook of the same size, which is equal to the number of typical source outputs. The details of this achievability proof are given in Appendix \ref{ss:sliding}.

The second coding scheme, which was studied in \cite{Gunduz:ITW:07} for a single relay channel, uses explicit binning at the source encoder and channel codes of different sizes for each terminal in the network. We call this scheme semi-regular encoding with backward decoding. The source is compressed (by binning) for each separate side information signal in the network, and hence a different rate of information is transmitted to each user; however, the rate of the channel codes for the terminals that have already decoded the message and are cooperating to forward it to the next terminal are the same. This is why we call this coding scheme a semi-regular encoding scheme. For decoding we use nested backward decoding \cite{Sikora:IT:06}, \cite{Kramer:Allerton:05}.  The detailed analysis of this achievability scheme is given in Appendix \ref{ss:backward}.

These two coding schemes essentially differ in terms of the delay in decoding, while each message block is decoded after a delay of $K$ channel blocks in the case of sliding-window decoding, the delay is much larger in the case of backward decoding, since the destination can start decoding only after receiving all the channel blocks. In the case of pure channel coding, the two schemes have exactly the same encoding structure; hence, the tradeoff is between the delay and the complexity. However; in the joint source-channel coding setting, there is another difference between the proposed coding techniques. While backward decoding works together with explicit binning, sliding-window decoding is based on matching the source outcome directly to a channel input. Therefore, in practical systems, the backward decoding scheme can be directly implemented using the existing point-to-point source codes and DF channel codes for the relay channel. However, the sliding-window coding scheme requires building new codes that implement the joint source-channel encoding and decoding techniques in practice.

In the following theorem it is shown that the proposed schemes achieve the source-channel capacity in a physically degraded relay network with degraded side information sequences. The definition of a physically degraded relay network is given below.

\begin{defn}\label{d:degraded}
A discrete memoryless relay network is said to be \emph{physically degraded} if
\begin{align}
p(y_{i+1},\ldots,y_{K+1} & |y_i,x_0,\ldots,x_{K}) = p(y_{i+1},\ldots,y_{K+1}|y_i,x_i,\ldots,x_{K})
\end{align}
for all $i=1,\ldots,K$, or equivalently if
\[(X_0, \ldots, X_{i-1}) \rightarrow (Y_i, X_i, \ldots,X_{K}) \rightarrow (Y_{i+1},\ldots, Y_{K+1}) \]
forms a Markov chain for all $i=1,\ldots,K$.
\end{defn}

\begin{thm}\label{t:degraded}
For a physically degraded relay network in which the side information sequences also form a Markov chain in the same order, i.e.,
\[S_0 \rightarrow S_1 \rightarrow \cdots \rightarrow S_{K+1}, \]
the source-channel capacity is given by
\begin{align}
   \sup_{p(x_0, x_1, \ldots, x_K)} \min_{i=1,\ldots, K+1} \frac{I(X_0^{i-1};Y_i |X_i^{K+1})}{H(S_0|S_i)}. \label{pdeg_cond}
\end{align}
\vspace{.1in}
\end{thm}

\begin{proof}
The converse for degraded relay networks follow from the cut-set bound. Consider the set $\mathcal{S}_i = \{T_0, \ldots, T_{i-1}\}$ and assume that the terminals in $\mathcal{S}_i$ all have access to the source vector $S_0^m$; hence, they can cooperate perfectly for transmitting $S_0^m$. We further assume that the remaining terminals can also cooperate perfectly by pooling all the available side information vectors $S_i^m, \ldots, S_{K+1}^m$ as well as their received channel outputs $Y_i^m, \ldots, Y_{K+1}^m$. This reduces to a point-to-point scenario for which the following is a necessary condition for reliable transmission:
\begin{align*}
   H(S_0|S_i, \ldots, S_{K+1}) &\leq r I(X_0^{i-1}; Y_i^{K+1}|X_i^{K+1}).
\end{align*}
From the degradedness assumption of the side information vectors we have
\begin{align*}
   H(S_0|S_i, \ldots, S_{K+1}) &=  H(S_0|S_i),
\end{align*}
and from the physically degraded channel assumption we have
\begin{align*}
   I(X_0^{i-1}; Y_i^{K+1}|X_i^{K+1}) = I(X_0^{i-1}; Y_i|X_i^{K+1}).
\end{align*}
We complete the proof of the theorem by considering all $\mathcal{S}_i$ sets corresponding to $i=1, \ldots, K+1$.
\end{proof}

\section{Extension to Relay-Broadcast Networks}\label{s:relay_bc}

\begin{figure*}
\centering
\psfrag{X0}{$X_0$}\psfrag{X1}{$X_1$}\psfrag{Xk}{$X_K$}\psfrag{Xk1}{$X_{K+1}$}\psfrag{Xkm}{$X_{K+L}$}
\psfrag{Y1}{$Y_1$}\psfrag{Yk}{$Y_K$}\psfrag{Yk1}{$Y_{K+1}$}\psfrag{Ykm}{$Y_{K+L}$}
\psfrag{S}{$S$}\psfrag{S1}{$S_1$}\psfrag{Sk}{$S_K$}\psfrag{Sk1}{$S_{K+1}$}\psfrag{Skm}{$S_{K+L}$}
\psfrag{hSk1}{$\hat{S}_{K+1}$}\psfrag{hSkm}{$\hat{S}_{K+L}$}
\psfrag{Source}{$\mathrm{Source}$} \psfrag{R1}{$~~~\mathrm{T_1}$}\psfrag{Rk}{$~~~\mathrm{T_K}$}
\psfrag{D1}{$\mathrm{~~T_{K+1}}$}\psfrag{Dm}{$\mathrm{~T_{K+L}}$}
\psfrag{Ch}{$p(y_1, \ldots, y_{K+L}|x_0, x_1, \ldots, x_{K+L})$}
\includegraphics[width=5in]{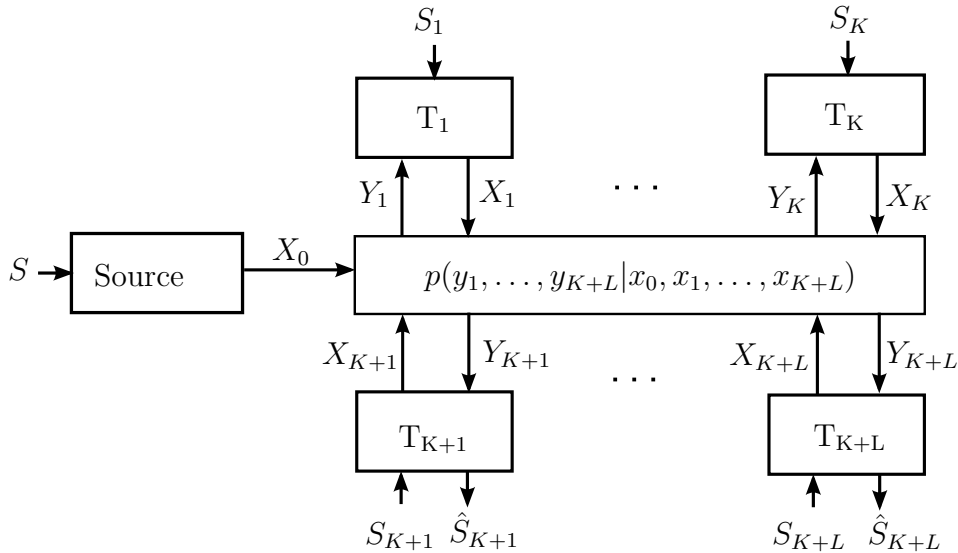}
\caption{Transmission of a discrete memoryless source over a relay-broadcast network with $K$ relay terminals and $L$ destination terminals. } \label{f:rbn}
\end{figure*}

It is possible to generalize the achievability results in Section \ref{s:results} to a relay-broadcast network, in which there are multiple receivers interested in decoding the source samples in a lossless fashion (see Figure \ref{f:rbn}). In this setting, terminals $T_{K+1}, \ldots, T_{K+L}$ not only decode the source sequence, but can also relay the decoded sequence for each other. We call this more general network a \emph{relay-broadcast network}. The following theorem provides an achievable source-channel code rate for a relay-broadcast network. The proof follows similarly to the proof of Theorem \ref{t:suff_cond}, in which a subset of the relays and the destination terminals apply the DF protocol in some given order.

In this model, the DM channel is characterized by the conditional distribution
\begin{align*}
p & (y_1^n, \ldots, y_{K+L}^n|x_0^n, x_1^n, \ldots, x_{K+L}^n).
\end{align*}
The source and side information sequences are DM and characterized by the joint distribution
\[
p(s_0,\ldots,s_{K+1})
\]
over a finite alphabet $\mathcal{S}_0 \times
\cdots \times \mathcal{S}_{K+L}$. Terminal $T_0$ maps its observation $S_0^m$ to a channel
codeword of length-$n$ by the encoding function
$f_0^{(m,n)}:\mathcal{S}^m \rightarrow \mathcal{X}_0^n$. The
terminal $T_i$ has encoding functions $f_i^{(m,n)} = \{ f_{i,1}^{(m,n)},
\ldots, f_{i,n}^{(m,n)}\}$ such that $X_{i,t} = f_{i,t}^{(m,n)}(Y_{i,1}, \ldots, Y_{i, t-1}, S_i^m)$, for $1\leq i\leq K+L$ and $1\leq t \leq n$.

The decoder at a destination terminal $T_j$, $j=K+1, \ldots, K+L$, maps the channel output
$Y_j^n$ and its side information $S_j^m$ to an estimate
$\hat{S}_j^m$ by the decoding function $\hat{S}_j^m = g^{(m,n)}_j(Y_j^n, S_j^m)$. The probability of error is defined as $P_e^{(m,n)} = Pr\left\{ \bigcup_{j=K+1}^{K+L} \{\hat{S}_j^m \neq S_0^m \} \right\}$. The definition of an achievable source-channel code rate is similar to Definition \ref{d:achieve} using this new probability of error definition.

\begin{thm}\label{t:suff_cond_rbn}
For the DM relay-broadcast network with correlated relay and destination side information as in Figure \ref{f:rbn}, the source-channel code rate $r$ is achievable if, for all $i= 1, \ldots, N$,
\begin{align}\label{cond:relbrnet}
     r < \frac{I(X_{\pi(0:i-1)};Y_{\pi(i)} |X_{\pi(i)}, \ldots,X_{\pi(K+L)})}{H(S_0|S_{\pi(i)})},
\end{align}
for an $N$ satisfying $L \leq N \leq K+L$, and for some injection $\pi(\cdot)$ from $\{0,\ldots, N\}$ to $\{0,\ldots, K+L\}$ such that $\pi(0)=0$ and $\{K+1, K+2, \ldots, K+L\} \subseteq \pi(1:N)$, and some input distribution $p(x_0, \ldots, x_{K+L})$.
\vspace{.1in}
\end{thm}

When all the terminals in the network wish to decode the source sequence, i.e., $K=0$ in the relay-broadcast network, and the channel inputs of these terminals are deterministic, i.e., $|\mathcal{X}_{K+1}| = \ldots = |\mathcal{X}_{K+L}| = 1$, then the problem reduces to the problem of broadcasting a common source to multiple terminals each with different side information. This problem is studied in \cite{Tuncel:IT:06} and it is shown that the source-channel capacity in this setup is given by
\begin{align}
   \sup_{p(x_0)} \min_{i=1,\ldots, L} \frac{I(X_0;Y_i)}{H(S_0|S_i)}. \label{broadcast_cap}
\end{align}
For this special case, our achievable scheme based on regular encoding and sliding-window decoding reduces to the coding scheme in \cite{Tuncel:IT:06}, whereas the semi-regular encoding and backward decoding scheme, introduced in the proof of Theorem \ref{t:suff_cond} provides an alternative optimal coding scheme for the broadcasting problem.

Another special case for which the source-channel capacity can be completely characterized is when all the terminals are interested in receiving the source sequence and there is only a single terminal with transmission capability, that is, $K=0$ and $|\mathcal{X}_2| = \cdots =|\mathcal{X}_L| = 1$. We have the following result for this setup.

\begin{lem}\label{l:single_relay}
For the DM single-relay broadcast network with correlated relay and destination side information at the terminals, i.e., $K=0$, $L \geq 2$ and $|\mathcal{X}_2| = \cdots = |\mathcal{X}_L| = 1$, the source-channel capacity is given by
\begin{align}
   \sup_{p(x_0, x_1)} \min \left\{ \frac{I(X_0;Y_1 |X_1)}{H(S_0|S_1)}, \frac{I(X_0, X_1;Y_2)} {H(S_0|S_2)}, \ldots, \frac{I(X_0, X_1;Y_L)} {H(S_0|S_L)} \right\}. \label{c:lem}
\end{align}
\vspace{.1in}
\end{lem}

\begin{proof}
While the achievability is a direct result of Theorem \ref{t:suff_cond_rbn}, the converse follows from the cut-set bound. The first term in the minimization comes from the cut around terminal $T_1$, and the following terms follow from the cuts around each of the terminals $T_2, \ldots, T_L$.
\end{proof}

\section{Conclusion}\label{s:conc}

We have considered the reliable transmission of a discrete memoryless source signal over a cooperative multiple-relay relay-broadcast network in which the relays and the destinations all have access to a different side information signal correlated with the source signal. We have defined the source-channel code rate of a code for such a system as the number of source symbols per channel use that can be transmitted reliably to the destination. The goal is to identify the maximum source-channel code rate, which we have defined as the source-channel capacity.

We have developed two different joint source-channel cooperation schemes which generalize decode-and-forward relaying to joint source-channel cooperation in multiple-relay networks, and we have identified the achievable source-channel code rate for these schemes. The first scheme does not use explicit source binning and is based on regular block-Markov encoding with sliding-window joint source-channel decoding, while the second scheme applies explicit source binning and uses separate source and channel decoders based on semi-regular block-Markov encoding and backward decoding. Our schemes illustrate that these two encoding schemes, apart from leading to different amounts of delay in decoding, also require two different types of implementation in the joint source-channel coding context. The consequences of this variation on the design of the practical codes is the subject of future research. Moreover, we have proven the optimality of DF relaying in the joint source-channel setting for a physically degraded relay channel with degraded side information, and for a relay broadcast network when all the terminals are interested in decoding the source signal, but at most one of the destinations has the transmission capability.

\appendices

%

\section{Proof of Theorem \ref{t:suff_cond}: Regular encoding and sliding-window decoding}\label{ss:sliding}

For a given $(n, m)$ pair, we build the code as follows. Consider $B-K \geq 0$ source blocks, each consisting of $m$ source samples, enumerated as $S_0^m(b)$, $b=1,\ldots, B-K$. These source blocks will be transmitted to the destination $T_{K+1}$ over $B$ channel blocks, each consisting of $n$ channel uses. This corresponds to a source-channel code rate of $Bn/(B-K)m$ channel uses per source sample. Note that this rate gets arbitrarily close to $n/m$ as $B \rightarrow \infty$.

Without loss of generality, we consider the special injection $\pi(i)=i$ for $i=0,\ldots,K+1$, and present the achievable scheme for this permutation. Generalization to any other subset of the relays and any other decoding order follows similarly.

Fix $p(x_0, ..., x_{K})$ such that (\ref{cond1}) holds. We use superposition block Markov encoding and sliding-window decoding.

\emph{Source code generation: } Generate at random $M = 2^{m(H(S_0)+\epsilon)}$ i.i.d. source codewords $s^m(w_0)$ in $\mathcal{S}_0^m$, $w_0 \in [1,M]$, each drawn according to the distribution $\prod_{t=1}^m p(s_{0, t})$. This constitutes the source codebook.

\emph{Channel code generation: } Generate at random $M$ i.i.d. channel codewords $x^n_{K-1}(w_{K-1})$ in $\mathcal{X}_{K-1}^n$, $w_{K-1} \in [1,M]$, each drawn according to the distribution $\prod_{t=1}^n p(x_{K-1, t})$. This constitutes the random channel codebook of relay $T_{K-1}$.

Then for each $x^n_{K-1}(w_{K-1})$, generate at random $M$ conditionally i.i.d. channel codewords $x^n_{K-2}(w_{K-2}|w_{K-1})$, $w_{K-2} \in [1,M]$, each drawn according to the distribution $\prod_{t=1}^n p(x_{K-2, t} |x_{K-1, t}(w_{K-1}))$. This constitutes the random channel codebook of relay $T_{K-2}$.

We continue the generation of codebooks sequentially for the terminals $T_{K-3}$, $T_{K-4}$, \ldots, $T_0$. For each tuple of
\[
\{x_{i+1}^n(w_{i+1}|w_{i+2}, \ldots, w_{K-1}), x_{i+2}^n(w_{i+2}|w_{i+3}, \ldots, w_{K-1}), \ldots, x_{K-1}^n(w_{K-1}) \}
\]
generate $M$ conditionally independent channel codewords $x_i^n(w_i|w_{i+1}, \ldots, w_{K-1})$, $w_i \in [1,M]$, each drawn according to the distribution $\prod_{t=1}^n p(x_{i,t} |x_{i+1, t}(w_{i+1}, \ldots, w_{K-1}), \ldots, x_{K-1}^n(w_{K-1}) )$. This constitutes the random channel codebook of terminal $T_i$, for $i=K-3, K-4, \ldots, 0$.

Finally, the channel codebook generation process is repeated independently $K-1$ times, and these codebooks are used sequentially over different channel blocks, so that when the joint typicality decoding is applied simultaneously over $K$ consecutive channel blocks, the decoding errors corresponding to different blocks are independent of each other.

\begin{figure*}
\centering
\begin{tiny}
\begin{tabular}{|c|c|c|c|c|c|c|c|}
  \hline
  Terminal & Block $1$ & Block 2 & Block 3 \\ \hline
  $T_0$ & $x_0^n(w_0(1)|1, 1)$ & $x_0^n(w_0(2)|w_0(1), 1)$ & $x_0^n(w_0(3)|w_0(2), w_0(1))$ \\ \hline
  $T_1$ & $x_1^n(1|1)$ & $x_1^n(\hat{w}_0^1(1)|1)$ & $x_1^n(\hat{w}_0^1(2)|\hat{w}_0^1(1))$ \\ \hline
  $T_2$ & $x_2^n(1)$ & $x_2^n(1)$ & $x_2^n(\hat{w}_0^2(1))$ \\ \hline \hline
  Terminal & Block $B-2$ & Block B-1 & Block $B$ \\ \hline
  $T_0$ & $x_0^n(w_0(B-2)|w_0(B-3), w_0(B-4))$ & $x_0^n(1|w_0(B-2), w_0(B-3))$ & $x_0^n(1|1, w_0(B-2))$ \\ \hline
  $T_1$ & $x_1^n(\hat{w}_0^1(B-3)|\hat{w}_0^1(B-4))$ & $x_1^n(\hat{w}_0^1(B-2)|\hat{w}_0^1(B-3))$ & $x_1^n(1|\hat{w}_0^1(B-2))$ \\ \hline
  $T_2$  & $x_2^n(\hat{w}_0^2(B-4))$ & $x_2^n(\hat{w}_0^2(B-3))$ & $x_2^n(\hat{w}_0^2(B-2))$ \\ \hline
\end{tabular}
\caption{Channel codeword assignment for the regular encoding and sliding-window decoding scheme with $K=2$ relays. We transmit $B-2$ source blocks to the destination in $B$ channel blocks.} \label{f:table_sliding}
\end{tiny}
\end{figure*}

\emph{Encoding:} At channel block $b$, for $b=1, \ldots, B-K$, the source terminal $T_0$ finds the index $w_0(b)$ of the source outcome $s_0^m(b)$. The index $w_0(b)$ is set to $1$ if the source realization is not typical. We set $w_0(b)=1$ for $b>B-K$ for notational convenience.


From the decoding procedure, which will be presented next, at the beginning of block $b$ for $b=1, \ldots, B$, terminal $T_i$, $i = 0, \ldots, K$, has the estimates $\hat{w}_i(b-k+1)$ of $w_0(b-k+1)$ for $k \geq i+1$, where we let $\hat{w}_0(i)=w_0(i)$ for $i=1, \ldots,b$. Terminal $T_i$ then sends $x_i^n(\hat{w}_i(b-i)|\hat{w}_i(b-i-1), \ldots, \hat{w}_i(b-K))$ over channel block $b$ using the codebook whose turn has come, where we set $\hat{w}_i(b')=w_0(b')=1$ for every $b'<1$ and $b'>b-K+1$. See Fig. \ref{f:table_sliding} for an illustration of the encoding procedure in a network with $K=2$ relays.

\emph{Decoding: } At the end of block $b$, for $b=1, \ldots, B$, terminal $T_i$, $i=1, \ldots, K+1$, declares $\hat{w}_i(b-i+1) = w$ if there exists a unique index $w \in [1,M]$ for which
\begin{align*}
 (s^m(w), S^m_i(b-i+1)) \in A^m_\epsilon(S_0, S_i)
\end{align*}
and
\begin{align*}
(x_{i-1-j}^n(w|\hat{w}_i(b-i), \ldots, \hat{w}_i(b-j-K)), \ldots, &  x_{K-1}^n(\hat{w}_i(b-j-K)), Y_i^n(b-j)) \\
& \in A^n_\epsilon (X_{i-1-j}, \ldots, X_{K}, Y_i),
\end{align*}
in all the blocks $b-j$ for $j=0, \ldots, i-1$. An error is declared if no or more than one such index is found. Note that terminal $T_i$, $i=1, \ldots, K+1$, attempts decoding at channel blocks $i \leq b \leq B-K+i-1$.


\emph{Analysis of probability of error: } The probability of not having a typical source outcome vanishes as $m$ increases. Hence, in the rest of the error analysis we will focus on the analysis of error for decoding the typical source realization at the nodes.

Denote by $P_i(b)$ the probability of the event that a decoding error is made at terminal $T_i$, $i=1, \ldots, K+1$, in block $b$, $b=1, \ldots, B$, conditioned on the event that no decoding error is made in the previous blocks. That is, we have
\begin{align*}
P_i(b) \triangleq \mathrm{Pr} \{\hat{w}_i(b-i+1) \neq w_0(b-i+1) | E^c(b-1) \},
\end{align*}
where
\begin{align*}
E^c(b) \triangleq \{\hat{w}_i(b'-i+1) = w_0(b'-i+1) \mbox{ for all } b'=1, \ldots,b \mbox{ and } i=1, \ldots, K+1\}.
\end{align*}
Then the probability of error $P_e$ can be bounded as follows
\begin{align}
P_e &= \sum_{b=1}^B \mathrm{Pr} \{ \hat{w}_i(b-i+1) \neq w_0(b-i+1) \mbox{ for some } i \in \{1, \ldots, K\} | E^c(b-1) \} \cdot Pr \{ E^c(b-1) \}, \\
    & \leq \sum_{b=1}^B \sum_{i=1}^K P_i(b) \cdot Pr \{ E^c(b-1) \}.
\end{align}

Assuming no decoding error is made in the previous blocks, to calculate $P_i(b)$ we can assume
\[
\hat{w}_i(b'-i+1) = w_0(b'-i+1)
\]
for all $b'=1, \ldots, b-1$. Then terminal $T_i$, $i=1, \ldots, K$, declares $\hat{w}_i(b-i+1)=w$, for $b-i+1 \geq 1$, if $w$ is the unique index $w \in [1, M]$ such that
\begin{align}\label{joint_type_c1}
 & (s^m(w), S^m_i(b-i+1)) \in A^m_\epsilon(S_0, S_i)
\end{align}
and
\begin{align}\label{joint_type_c2}
(x_{i-1-j}^n(w|w(b-i), \ldots, w(b-j-K)),\ldots, x_{K-1}^n(w(b-j-K)), &   Y_i^n(b-j)) \nonumber \\
& \in A^n_\epsilon (X_{i-1-j}, \ldots, X_{K}, Y_i)
\end{align}
hold simultaneously for all the blocks $b-j$ for $j=0,1, \ldots, i-1$.

We define the following sets
\begin{align}
E^0_i(b) &\triangleq \{ w \in [1,M]: w \mbox{ satisfies } (\ref{joint_type_c1}) \}, \\
E_{i,j}(b) &\triangleq \{ w \in [1,M]: w \mbox{ satisfies } (\ref{joint_type_c2}) \}, \\
E_i(b)  & \triangleq \bigcap_{j=0}^{i-1} E_{i,j}(b),
\end{align}
for $i=1, \ldots, K$ and $b=1, \ldots, B$. Then, $P_i(b)$ can be written as
\begin{align}
P_i(b) &= \mathrm{Pr} \left\{ w_0(b-i+1) \notin E^0_i(b) \mbox{ or } w_0(b-i+1) \notin E_i(b) \nonumber \right. \\
    &~~~~~ \mbox{ or for some } w'\in E^0_i(b) \cap E_i(b)   \left.  \mbox{ but } s^m(w') \neq S_0^m(b-i+1) | E^c(b-1) \right\} \\
  & \leq \frac{1}{\mathrm{Pr} \{E^c(b-1) \}} \left[ \mathrm{Pr} \{ w_0(b-i+1) \notin E^0_i(b) \} \right. +  \mathrm{Pr} \{ w_0(b-i+1) \notin E_i(b) \} \nonumber \\
  & ~~~~~ \left. +  \mathrm{Pr} \{ \mbox{ for some } w' \in E^0_i(b) \cap E_i(b) \mbox{ but } s^m(w') \neq s_0^m(b-i+1) \} \right].
\end{align}
Hence, the probability of error can be bounded as
\begin{align}
P_e & \leq \sum_{b=1}^B \sum_{i=1}^K  \left[ \mathrm{Pr} \{ w_0(b-i+1) \notin E^0_i(b) \} \right. +  \mathrm{Pr} \{ w_0(b-i+1) \notin E_i(b) \} \nonumber \\
  & ~~~~~ \left. +  \mathrm{Pr} \{ \mbox{for some } w' \in E^0_i(b) \cap E_i(b) \mbox{ but } s^m(w') \neq s_0^m(b-i+1) \} \right].
\end{align}

The first two arguments in the above summation can be made arbitrarily small for large enough $m$ and $n$ \cite{Csiszar:book}. On the other hand, we have
\begin{align}
\mathrm{Pr} \{ \mbox{for some } w'\in E^0_i(b) \cap E_i(b) & \mbox{ but } s^m(w') \neq s_0^m(b-i+1) \}  \nonumber \\
 &\leq \sum_{w'\in [1,M], w' \neq w_0(b-i+1)} \mathrm{Pr}\{w' \in E^0_i(b) \cap E_i(b) \} \label{j_ineq_2} \\
 &= \sum_{w'\in [1,M], w' \neq w_0(b-i+1)} \mathrm{Pr} \{w' \in E^0_i(b)\} \cdot \mathrm{Pr} \{w' \in E_i(b) \}   \label{j_ineq_3} \\
 &= \sum_{w'\in [1,M], w' \neq w_0(b-i+1)} \mathrm{Pr} \{w' \in E^0_i(b)\} \cdot \prod_{j=0}^{i-1} \mathrm{Pr} \{w' \in E_{i,j}(b) \},   \label{j_ineq_4} \\
 &\leq (M-1) 2^{-m(I(S_0; S_i)-3\epsilon)} 2^{-n(I(X_0, \ldots, X_{i-1} ; Y_i |X_i, \ldots, X_{K-1})- 6i\epsilon)}   \label{j_ineq_5}
\end{align}
where (\ref{j_ineq_3}) follows from the independence among the source and the channel codebooks; (\ref{j_ineq_4}) follows from the independence among the channel codebooks at consecutive channel blocks; and (\ref{j_ineq_5}) follows from the chain rule of mutual information and the following inequality:
\[
\mathrm{Pr} \{w' \in E_{i,j}(b) \} \leq 2^{-n(I(X_{i-1-j}; Y_i |X_{i-j}, \ldots, X_{K-1})- 6\epsilon)}
\]
for $w' \neq w_0(b-i+1)$. Finally, substituting the value of $M$, we get
\begin{align}
\mathrm{Pr} \{ \mbox{for some } w'\in E^0_i(b) \cap E_i(b) & \mbox{ but } s^m(w') \neq s_0^m(b-i+1) \}  \nonumber \\
    &  \leq   2^{-m[H(S_0|S_i) - bI(X_0, \ldots, X_{i-1} ; Y_i |X_i, \ldots, X_{K-1})- \epsilon']}  \label{j_ineq_6},
\end{align}
where $\epsilon' \triangleq (4 + 6bi)\epsilon$.

For sources and channels satisfying the conditions of the theorem, by appropriately choosing $\epsilon$ and letting $m,n \rightarrow \infty$, we can have an arbitrarily small probability of error.

\section{Proof of Theorem \ref{t:suff_cond}: Semi-regular encoding and backward decoding}\label{ss:backward}

In backward decoding for the single relay channel \cite{Willems:thesis}, while the relay decodes each message block right after it is transmitted as in the sliding-window decoding scheme in Section \ref{ss:sliding}, the destination waits until all message blocks are transmitted and decodes them in the reverse order by removing the interference from the decoded messages. In backward decoding, each message is decoded at the destination using the signal received in a single block as opposed to signal combining in sliding-window decoding; however, backward decoding introduces additional delay, which grows with the number of relays in the network.

We use the multiple-relay backward decoding scheme for the transmission over the channel \cite{Sikora:IT:06}, \cite{Kramer:Allerton:05}. This is a nested backward decoding scheme constructed recursively such that, in each step of the recursion, a new node decodes the messages up to that point using backward decoding.

Rather than the joint decoding approach in Section \ref{ss:sliding}, in the case of backward decoding we use separate source and channel encoders/decoders at each node. In the case of separate source-channel coding, the source samples need to be compressed separately for each node such that each node receives enough information to decode the source sequence when combined with its own side information sequence. For example, in the case of a single relay terminal, the source terminal generates two independent bin indices, one for the relay and one for the destination. The transmission of the bin indices require using channel codes at different rates. However, note that we need a different rate for each receiving node in the network, and all the terminals participating in the transmission of the source terminals to the same node can use the same rate for their codes.

Due to the nested structure of the backward decoding scheme its complexity increases quickly with the increasing number of relays. Hence, for simplicity, we present the transmission scheme for $K=2$ relays with a decoding order of $T_1, T_2, T_3$. Extension to the setup with more relays is a direct generalization.

As in Section \ref{ss:sliding}, for a given $(n, m)$ pair, we build the code as follows. Fix $p(x_0, x_1, x_2)$ such that (\ref{cond1}) holds. A total of $B^2m$ source samples will be transmitted over $(B+1)^2n$ channel uses. This corresponds to a source-channel code rate of $(B+1)^2n/B^2m$ which gets arbitrarily close to $n/m$ as $B \rightarrow \infty$.

\emph{Source code generation: } Corresponding to each terminal $T_i$, for $i=1,2,3$, we consider $M_i = 2^{mR_i}$ bins, called the $T_i$ bins. All possible source outcomes $s_0^m \in T_{S_0}^m$ are partitioned randomly and uniformly into these bins, independently for each side information sequence, i.e., the distribution into $M_i$ bins for $S_i$ is independent of the distribution into $M_j$ bins for $S_j$ for $i \neq j$. This bin assignment, which corresponds to source compression, is made available to all the terminals.

\emph{Channel code generation: } For the channel codebook, generate at random $M_3$ channel codewords $x_2^n(j_3)$ for $j_3 \in [1,M_3]$ i.i.d. with $p(x_2^n(j_3))=\Pi_{t=1}^n p(x_{2,t})$, and index them as $x_2^n(j_3)$ with $j_3 \in [1,M_3]$.

Then for each $x_2^n(j_3)$, generate at random $M_2$ conditionally independent channel codewords $x_{1}^n(j_2|j_3)$, $j_2 \in[1,M_2]$, with probability $p(x_{1}^n|x_2^n(j_3))=\Pi_{t=1}^n p(x_{1,t}|x_{2,t}(j_3))$, and index them as $x_1^n(j_2| j_3)$ with $j_2 \in [1,M_2]$.

Finally, generate at random the codebook of size $M_1$ for each possible combination of $(x_1^n(j_2|j_3), x_2^n(j_3))$, with probability $p(x_0^n|x_1^n(j_2|j_3), x_2^n(j_3))=\Pi_{t=1}^n p(x_{0,t}|x_{1,t}(j_2, j_3), x_{2,t}(j_3))$, and index them as $x_0^n(j_1|j_2, j_3)$ with $j_1 \in [1,M_1]$.

\begin{figure*}
\begin{tiny}
\begin{tabular}{|l|l|l|l|l|l|l|l|}
  \hline
  Block & $1$ & 2 & $\cdots$ & B & B+1 \\ \hline
  $T_0$ & $x_0^n(w_{1,1}|1, 1)$ & $x_0^n(w_{2,1}|w_{1,2}, 1)$ & $\cdots$ & $x_0^n(w_{B,1}|w_{B-1,2}, 1)$ & $x_0^n(1|w_{B,2}, 1)$ \\ \hline
  $T_1$ & $x_1^n(1|1)$ & $x_1^n(\hat{w}_{1,2}^1|1)$ & $\cdots$ & $x_1^n(\hat{w}_{B-1,2}^1|1)$ & $x_1^n(\hat{w}_{B,2}^1|1)$ \\ \hline
  $T_2$ & $x_2^n(1)$ & $x_2^n(1)$ & $\cdots$ & $x_2^n(1)$ & $x_2^n(1)$ \\ \hline \hline
  Block & $(B+1)+1$ & $(B+1)+2$ & $\cdots$ & $(B+1)+B$ & $(B+1)+B+1$ \\ \hline
  $T_0$ & $x_0^n(w_{B+1,1}|1, w_{1,3})$ & $x_0^n(w_{B+2,1}|w_{B+1,2}, w_{2,3})$ & $\cdots$ & $x_0^n(w_{2B,1}|w_{2B-1,2}, w_{B,3})$ & $x_0^n(1|w_{2B,2}, 1)$ \\ \hline
  $T_1$ & $x_1^n(1|\hat{w}_{1,3}^1)$ & $x_1^n(\hat{w}_{B+1,2}^1|\hat{w}_{2,3}^1)$ & $\cdots$ & $x_1^n(\hat{w}_{2B-1,2}^1|\hat{w}_{B,3}^1)$ & $x_1^n(\hat{w}_{2B,2}^1|1)$ \\ \hline
  $T_2$  & $x_2^n(\hat{w}_{1,3}^2)$ & $x_2^n(\hat{w}_{2,3}^2)$ & $\cdots$ & $x_2^n(\hat{w}_{B,3}^2)$ & $x_2^n(1)$ \\ \hline \hline
   & $\cdots$ & $\cdots$ & $\cdots$ & $\cdots$ & $\cdots$ \\ \hline \hline
  Block & $k(B+1)+1$ & $k(B+1)+2$ & $\cdots$ & $k(B+1)+B$ & $k(B+1)+B+1$ \\ \hline
  $T_0$ & $x_0^n(w_{kB+1,1}|1, w_{(k-1)B+1,3})$ & $x_0^n(w_{kB+2,1}|w_{kB+1,2}, w_{(k-1)B+2,3})$ & $\cdots$ & $x_0^n(w_{kB+B,1}|w_{kB+B-1,2}, w_{(k-1)B+B,3})$ & $x_0^n(1|w_{kB+B,2}, 1)$ \\ \hline
  $T_1$ & $x_1^n(1|\hat{w}_{(k-1)B+1,3}^1)$ & $x_1^n(\hat{w}_{kB+1,2}^1|\hat{w}_{(k-1)B+2,3}^1)$ & $\cdots$ & $x_1^n(\hat{w}_{kB+B-1,2}^1|\hat{w}_{(k-1)B+B,3}^1)$ & $x_1^n(\hat{w}_{kB+B,2}^1|1)$ \\ \hline
  $T_2$  & $x_2^n(\hat{w}_{(k-1)B+1,3}^2)$ & $x_2^n(\hat{w}_{(k-1)B+2,3}^2)$ & $\cdots$ & $x_2^n(\hat{w}_{(k-1)B+B,3}^2)$ & $x_2^n(1)$ \\ \hline \hline
  & $\cdots$ & $\cdots$ & $\cdots$ & $\cdots$ & $\cdots$ \\ \hline \hline
  Block & $B(B+1)+1$ & $B(B+1)+2$ & $\cdots$ & $B(B+1)+B$ & $B(B+1)+B+1$ \\   \hline
  $T_0$ & $x_0^n(1|1, w_{B(B-1)+1,3})$ & $x_0^n(1|1, w_{B(B-1)+2,3})$ & $\cdots$ & $x_0^n(1|1, w_{B^2,3})$ & $x_0^n(1|1, 1)$ \\ \hline
  $T_1$ & $x_1^n(1|\hat{w}_{B(B-1)+1,3}^1)$ & $x_1^n(1|\hat{w}_{B(B-1)+2,3}^1)$ & $\cdots$ & $x_1^n(1|\hat{w}_{B^2,3}^1)$ & $x_1^n(1|1)$ \\ \hline
  $T_2$  & $x_2^n(\hat{w}_{B(B-1)+1,3}^2)$ & $x_2^n(\hat{w}_{B(B-1)+2,3}^2)$ & $\cdots$ & $x_2^n(\hat{w}_{B^2,3}^2)$ & $x_2^n(1)$ \\ \hline
\end{tabular}
\caption{Channel codeword assignment for the semi-regular encoding and backward decoding scheme with $K=2$ relays. We transmit $B^2$ source blocks to the destination in $(B+1)^2$ channel blocks.} \label{f:table_backward}
\end{tiny}
\end{figure*}

\emph{Encoding: } Consider a source sequence $s_0^{B^2m}$ of length
$B^2m$. Partition this sequence into $B^2$ portions, $s_{0,b}^m$,
$b=1,\ldots,B^2$. Similarly, partition the side information sequences
into $B^2$ length-$m$ blocks $s_i^{B^2m}=[s_{i,1}^m,\ldots,s_{i,B^2}^m]$ for $i=1,2,3$. The bin index of the $j$th block of the source output sequence $s_{0,j}^m$ with respect to $T_i$ bins is denoted by $w_{j, i}$. The estimate of $w_{j,i}$ at node $k$, $k=1, \ldots, K+1$, is denoted by $\hat{w}_{j,i}^k$. See Fig. \ref{f:table_backward} for an illustration of the encoding scheme.

In block 1, $T_0$ observes $s_{0,1}^m$, and finds the corresponding bin index $w_{1,1} \in [1,M_1]$. It transmits the channel codeword $x_0^n(w_{1,1}|1, 1)$. The relays $T_1$ and $T_2$ simply transmit $x_1^n(1|1)$ and $x_2^n(1)$, respectively. In block $2$, $T_0$ transmits the channel codeword
$x_0^n(w_{2,1}|w_{1,2}, 1)$. The relays $T_1$ and $T_2$ transmit $x_1^n(\hat{w}_{1,2}^1|1)$ and $x_2^n(1)$, respectively, where $\hat{w}_{1,2}^1$ is the $T_2$ bin index of the estimate $\hat{s}_{1,1}^m$ at the relay $T_1$. In the following blocks $b=2,\ldots,B$, the source terminal transmits the channel codeword $x_0^n(w_{b,1}|w_{b-1,2}, 1)$ where $w_{b,i} \in [1,M_i]$ for $i=1,2$. In block $B+1$, $T_0$ transmits $x_0^n(1| w_{2,B}, 1)$.

The first relay $T_1$ estimates the source block $s_{0,b-1}^m$ at the end of block $b-1$, denoted by $\hat{s}_{1,b-1}^m$, and finds the corresponding $T_2$ bin index $\hat{w}_{b-1, 2}^1 \in [1,M_2]$. At block $b$, for $b=2,\ldots,B+1$, $T_1$ transmits the channel codeword $x_1^n(\hat{w}_{2,b-1}^1| 1)$. For the first $B+1$ channel blocks terminal $T_2$ transmits $x_2^n(1)$.

At the end of channel block $B+1$, the relay $T_2$ decodes the first source block $s_{0,1}^m$ by backward decoding. Having estimated the source blocks $s_{0,1}^m, \ldots, s_{0,B}^m$ by backward decoding, $T_2$ joins the transmission for forwarding the $T_3$ bin indices of the first source block to the destination.

In the following channel blocks $T_0$ superposes the transmission of the second source block to $T_1$ and $T_2$ on the transmission of the $T_3$ bin indices of the first source block to $T_3$. In the channel blocks from $B+1$ up to $2B+1$, the source terminal transmits, in order,
\[
x_0^n(w_{B+1,1}|1, w_{1,3}), \ldots, x_0^n(w_{2B,1}|w_{2B-1,2}, w_{B,3}), x_0^n(1|w_{2B,2}, 1).
\]
The relay $T_1$ transmits, in order,
\[
x_1^n(1|\hat{w}_{1,3}^1), x_1^n(\hat{w}_{B+1,2}^1|\hat{w}_{2,3}^1), \ldots, x_1^n(\hat{w}_{2B-1,2}^1|\hat{w}_{B,3}^1), x_1^n(\hat{w}_{2B,2}^1|1).
\]
Finally, the second relay $T_2$ transmits $x_2^n(\hat{w}_{1,3}^2), \ldots, x_2^n(\hat{w}_{B,3}^2, x_2^n(1))$.

They continue similarly for a total of $B$ channel blocks of $Bn$ channel uses each. In each of these $B$ groups of channel blocks $Bm$ source samples are encoded by terminal $T_0$. In the last block of $(B+1)n$ channel uses, no new source samples are encoded. Terminal $T_0$ transmits
\[
x_0^n(1|1, w_{B(B-1)+1,3}), \ldots, x_0^n(1|1, w_{B^2,3}), x_0^n(1|1,1),
\]
while $T_1$ transmits $x_1^n(1|\hat{w}_{B(B-1)+1,3}^1), \ldots, x_1^n(1|\hat{w}_{B^2,3}^1), x_1^n(1|1)$, and having estimated the $T_3$ bin indices for the last $Bm$ source samples, $T_2$ transmits $x_2^n(\hat{w}_{B(B-1)+1,3}^2), \ldots, x_2^n(\hat{w}_{B^2,3}^2), x_2^n(1)$. It can be noted that the last channel block of $n$ channel uses is unused, but is included to simplify the expressions as it does not reduce the source-channel code rate in the limit of an infinite number of source and channel blocks.

\emph{Decoding and error probability analysis:} Since each node tries to decode each block of the source sample, we denote the estimate of source block $s_{0,b}^m$ at note $T_k$, $k=1, \ldots, K+1$, by $\hat{s}_{1,b}^m$.

The relay $T_1$ decodes the source signal by sequentially reconstructing source block $s_{0,b}^m$ at the end of the channel block, in which the corresponding $T_1$ bin index is transmitted by the source terminal $T_0$.

Consider channel block $k(B+1)+b$ for $k=0,1, \ldots, B-1$ and $b=1,\ldots,B$. Assume that at the end of block $k(B+1)+b-1$, $T_1$ has estimated the source blocks $s_{0,1}^m, \ldots, s_{0,kB+b-1}^m$ correctly, i.e., $\hat{s}_{1,b}^m = s_{0,b}^m$ for $b=1, \ldots, kB+b-1$. Hence, it can correctly find the $T_2$ bin index for $s_{0,kB+b-1}^m$ and the $T_3$ bin index for $s_{0,(k-1)B+b}^m$, i.e., $\hat{w}_{kB+b-1, 2}^1 = w_{kB+b-1, 2}$ and $\hat{w}_{(k-1)B+b, 3}^1 = w_{(k-1)B+b, 3}$. Using this information and its received signal $y_1^n(k(B+1)+b)$, the $T_1$ channel decoder attempts to decode $w_{kB+b, 1}$, i.e., the $T_1$ bin index corresponding to $s_{0,kB+b}^m$, by looking for a unique index $w$ such that
\begin{align*}
(x_0^n(w|w_{kB+b-1, 2}, w_{(k-1)B+b, 3}), x_1^n(\hat{w}_{kB+b-1, 2}^1| \hat{w}^1_{(k-1)B+b, 3}), x_2^n(\hat{w}^2_{(k-1)B+b, 3}), Y_1^n(b)) \in A^n_\epsilon (X_1, X_2, X_3, Y_1),
\end{align*}
where $\hat{w}_{b, 3}^i = w_{b,3} = 1$ for $i=1,2$, if $b<1$. If such a unique index is found, then it is assigned as the bin index estimation of the source sequence $s_{0,b}^m$ at node $T_1$, $w_{b,1}^1$.

The estimated bin index is then given to the $T_1$ source decoder. With the $T_1$ bin index and the side information $s_{1,kB+b}^m$, the $T_1$ source decoder estimates $s_{0,kB+b}^m$ by looking for a unique source codeword which is jointly typical with $S_1^m(kB+b)$ and whose bin index is $\hat{w}_{kB+b,1}^1$. If such a unique codeword is found, it is assigned as the source estimation at $T_1$, denoted by $\hat{s}_{1,b}^m$.

At the end of channel block $k(B+1)+B+1$, for $k=0,\ldots, B-1$, the relay $T_2$ decodes the $T_2$ bin index for the source block $s_{0,kB+B}^m$. It then continues decoding in the reverse order by backward decoding. Assuming that it has estimated the source blocks $s_{0,1}^m, \ldots, s_{0,(k-1)B+B}^m$ and $s_{0,kB+b+1}^m, \ldots, s_{0,kB+B}^m$ correctly, it can find the $T_1$ bin index of $s_{0,kB+b+1}^m$ and the $T_3$ bin index of $s_{0,(k-1)B+b}^m$. Using this information and its received signal $Y_2^n(k(B+1)+b)$, the $T_2$ channel decoder attempts to decode $w_{kB+b, 2}$ by looking for a unique index $w$ such that
\begin{align*}
(x_0^n(w_{kB+b+1, 1}|w, w_{(k-1)B+b, 3}), x_1^n(w| \hat{w}^1_{(k-1)B+b, 3}), x_2^n(\hat{w}^2_{(k-1)B+b, 3}), & Y_2^n(k(B+1)+b)) \\
& \in A^n_\epsilon (X_1, X_2, X_3, Y_2).
\end{align*}
If such a unique index is found, then it is assigned as the bin index estimation of the source sequence $s_0^m(kB+b)$ at node $T_2$, $\hat{w}_{kB+b, 2}$.

The estimated bin index is then given to the $T_2$ source decoder. With the $T_2$ bin index and the side information $s_{2,kB+b}^m$, the $T_2$ source decoder estimates $s_{0,kB+b}^m$ by looking for a unique source codeword which is jointly typical with $S_2^m(kB+b)$ and whose bin index is $\hat{w}_{kB+b,2}^2$. If such a unique codeword is found, it is assigned as the source estimation at $T_1$, denoted by $\hat{s}_{2,kB+b}^m$.


Decoding at the destination node $T_3$ is also done using backward decoding, but the destination waits till the end of channel block $B(B+1)+B$. It
first tries to decode the $T_3$ bin index of the last source block $s_{0,B^2}^m$ using the received signal at channel block $B(B+1)+B$. Consider decoding of $s_0^m((k-1)B+b)$ at channel block $k(B+1)+b$ for $k=1, \ldots, B$ and $b=1, \ldots, B$. Assuming that it has decoded the source blocks $s_{0,(k-1)B+b+1}^m, \ldots, s_{0,B^2}^m$ correctly, it can find the $T_1$ bin index of $s_{0,kB+b}^m$ and the $T_2$ bin index of $s_{0,kB+b-1}^m$. Using these information and its received signal $Y_3^n(k(B+1)+b)$, the $T_3$ channel decoder attempts to decode $w_{(k-1)B+b, 3}$ by looking for a unique index $w$ such that
\begin{align*}
(x_0^n(w_{kB+b, 1}| w_{kB+b-1, 2}, w), x_1^n(\hat{w}^1_{kB+b-1, 2}| w), x_2^n(w), & Y_3^n(k(B+1)+b)) \\
& \in A^n_\epsilon (X_1, X_2, X_3, Y_3).
\end{align*}
If such a unique index is found, then it is assigned as the bin index estimation of the source sequence $s_0^m((k-1)B+b)$ at node $T_2$, $\hat{w}_{kB+b, 2}$.

The estimated bin index is then given to the $T_3$ source decoder. With the $T_3$ bin index and the side information $s_{3,(k-1)B+b}^m$, the $T_3$ source decoder estimates $s_{0,(k-1)B+b}^m$ by looking for a unique source codeword that is jointly typical with $S_3^m((k-1)B+b)$ and whose bin index is $\hat{w}_{(k-1)B+b,3}^3$.


We upper bound the probability of error by counting errors made not only at the destination, but at any terminal in the network.
\begin{eqnarray}
  P_e & \leq & \sum_{s^{B^2m}} \mathrm{Pr} \left\{ \bigcup_{b=1, \ldots, B^2} \bigcup_{k=1, \ldots, K+1} \{\hat{S}_{b,k}^m \neq s_{0,b}^m | S^{B^2m} = s^{B^2m} \} \right\}  Pr \left\{ S^{B^2m} = s^{B^2m}  \right\}.
\end{eqnarray}
We can write the above union in a recursive manner by considering the decoding order of the backward decoding scheme; that is, the above event of having an error in estimating any of the source blocks at any of the nodes can be written as the union of error events, such that each corresponds to an error event at one stage of decoding given that no error has been made previously. Note that the decoding order is as follows: $\hat{s}_{1,1}^m, \hat{s}_{1,2}^m, \ldots, \hat{s}_{1,B}^m$, $\hat{s}_{2,B}^m, \ldots, \hat{s}_{2,1}^m$, $\hat{s}_{1,B+1}^m, \ldots, \hat{s}_{1,2B}^m, \hat{s}_{2,2B}^m, \ldots, \hat{s}_{2,B+1}^m, \ldots$, $\hat{s}_{1,(B-1)B+1}^m, \ldots, \hat{s}_{1,B^2}^m$, $\hat{s}_{2,B^2}^m, \ldots, \hat{s}_{2,(B-1)B+1}^m$, $\hat{s}_{3,B^2}^m, \hat{s}_{3,B^2-1}^m,\ldots, \hat{s}_{3,1}^m$. We have
\begin{eqnarray}
  P_e & \leq & \sum_{s^{B^2m}} \left\{ \sum_{b=1, \ldots, B^2} \sum_{k=1, \ldots, K+1} \mathrm{Pr} \{\hat{S}_{b,k}^m \neq s_{0,b}^m | E_{b,k}, S^{B^2m} = s^{B^2m} \}  \nonumber \right. \\
  && \left. \cdot \mathrm{Pr}\{E_{b,k} | S^{B^2m} = s^{B^2m} \} \right\}  Pr \left\{ S^{B^2m} = s^{B^2m} \right\},
\end{eqnarray}
where we define $E_{b,k}$ as the event that all previous estimations are correct when $T_k$ estimates $S_{0,b}^m$. Then we get
\begin{eqnarray}
  P_e & \leq & \sum_{s^{B^2m}} \left\{ \sum_{b} \sum_{k} \mathrm{Pr} \{\hat{S}_{b,k}^m \neq s_{0,b}^m | E_{b,k}, S^{B^2m} = s^{B^2m} \} \right\} \cdot Pr \left\{ S^{B^2m} = s^{B^2m} \right\}, \\
  & \leq & \sum_{s^{B^2m}} \left\{ \sum_{b} \sum_{k} \mathrm{Pr} \{\hat{S}_{b,k}^m \neq s_{0,b}^m,  \hat{w}_{b,k}^k = w(b,k) | E_{b,k}, S^{B^2m} = s^{B^2m} \} \right. \nonumber \\
  & & \left. + \mathrm{Pr} \{\hat{S}_{b,k}^m \neq s_{0,b}^m,  \hat{w}_{b,k}^k \neq w(b,k) |  E_{b,k}, S^{B^2m} = s^{B^2m} \} \right\} \cdot Pr \left\{ S^{B^2m} = s^{B^2m} \right\}, \\
  & \leq & \sum_{s^{B^2m}} \left\{ \sum_{b} \sum_{k} \mathrm{Pr} \{\hat{S}_{b,k}^m \neq s_{0,b}^m |  \hat{w}_{b,k}^k = w(b,k), E_{b,k}, S^{B^2m} = s^{B^2m} \} \right. \nonumber \\
  & & \left. + \mathrm{Pr} \{\hat{w}_{b,k}^k \neq w(b,k) |  E_{b,k}, S^{B^2m} = s^{B^2m} \} \right\} \cdot Pr \left\{ S^{B^2m} = s^{B^2m} \right\}, \\
  & = &  \sum_{b} \sum_{k} \mathrm{Pr} \{\hat{S}_{b,k}^m \neq s_{0,b}^m |  \hat{w}_{b,k}^k = w(b,k), E_{b,k} \} \nonumber \\
  & &  + \sum_{s^{B^2m}} \left\{ \sum_{b} \sum_{k} \mathrm{Pr} \{\hat{w}_{b,k}^k \neq w(b,k) |  E_{b,k}, S^{B^2m} = s^{B^2m} \} \right\}  Pr \left\{ S^{B^2m} = s^{B^2m} \right\}, \\
  & = &  \sum_{b} \sum_{k} \mathrm{Pr} \{\hat{S}_{b,k}^m \neq s_{0,b}^m |  \hat{w}_{b,k}^k = w(b,k), E_{b,k} \} \nonumber \\
  & &  + \sum_{b}  \sum_{k} \sum_{w(b,k) \in \{1, \ldots, 2^{mR_k}\}} \mathrm{Pr} \{\hat{w}_{b,k}^k \neq w(b,k) |  E_{b,k}, W_{b,k} = w(b,k) \} \mathrm{Pr} \{W_{b,k} = w(b,k)\}, \\
  & = &  \sum_{b} \sum_{k} \mathrm{Pr} \{\hat{S}_{b,k}^m \neq s_{0,b}^m |  \hat{w}_{b,k}^k = w(b,k), E_{b,k} \} \nonumber \\
  & &  + \sum_{b}  \sum_{k} \frac{1}{2^{mR_k}} \sum_{w(b,k) \in \{1, \ldots, 2^{mR_k}\}} \mathrm{Pr} \{\hat{w}_{b,k}^k \neq w(b,k) |  E_{b,k}, W_{b,k} = w(b,k) \}. \label{final_eq}
\end{eqnarray}
Now, note in (\ref{final_eq}) that each term in the first summation corresponds to the error event at the source decoder of $T_k$ given that it is provided with the correct bin index, and each term in the second summation corresponds to the error event at the channel decoder of $T_k$, both conditioned on the fact that all the estimations up to that instant are correct. Following the usual arguments, we get
\begin{align}
  P_{e,i}(b) \leq &  \sum_{b} \sum_{k} \left \{ 2^{mR_k} \cdot 2^{m(H(S_0|S_k) + \epsilon)} + 2^{mR_k} \cdot 2^{-n(I(X_0^{k-1} ; Y_k| X_k^{K+1}) + \epsilon)} \right\}
\end{align}
Setting $R_k = H(S_0|S_k)$, and letting $m,n \rightarrow \infty$ while $\epsilon \rightarrow 0$, the error probability vanishes under the assumptions of the theorem.

\bibliographystyle{ieeetran}
\bibliography{ref3}

\end{document}